\documentclass[12t,conference,onecolumn]{IEEEtran}  %draftcls,onecolumn
\usepackage{times}
\usepackage{hyperref} % for links
\usepackage{amsmath}
\usepackage{amssymb}
\usepackage{amsthm}

\newtheorem{theorem}{Theorem}
\newtheorem{lemma}{Lemma}

\theoremstyle{definition}
\newtheorem{definition}{Definition}
\theoremstyle{remark}
\newtheorem{remark}{Remark}

\newcommand{\E}[1]{\mathbb{E} \left( #1 \right)}
\newcommand{\BRA}[1]{\left( #1 \right)}
\newcommand{\BRAs}[1]{\left\{ #1 \right \}}
\newcommand{\PR}[1]{Pr\left\{ #1 \right\}}
\newcommand{\bW}{{\textbf W}}
\newcommand{\bS}{{\textbf S}}
\newcommand{\bX}{{\textbf X}}
\newcommand{\bY}{{\textbf Y}}
\newcommand{\cA}{{\mathcal A}}
\newcommand{\cX}{{\mathcal X}}
\newcommand{\cY}{{\mathcal Y}}
\newcommand{\cS}{{\mathcal S}}
\newcommand{\chS}{{\hat{\mathcal S}}}

\newcommand{\ED}{\mathbb{E}(D)}

\newcommand{\Selmnt}{\textbf{s}_n}
\newcommand{\Selmnto}{\textbf{s}_n^1}
\newcommand{\Selmntt}{\textbf{s}_n^2}
\newcommand{\Shelmnt}{\hat{\textbf{s}}_n}
\newcommand{\calX}{{$\mathcal X$}}

\newcommand{\ie}{{\emph{i.e.}}}
\newcommand{\eg}{{\emph{e.g.}}}

\newcommand{\verdu}{ Verd{\'u} }

\newcommand{\limtoinf}{\lim_{n \to \infty }}
\newcommand{\liminftoinf}{\liminf_{n \to \infty }}
\newcommand{\limsuptoinf}{\limsup_{n \to \infty }}
\newcommand{\toinf} {\underset{n \rightarrow \infty}{\rightarrow}}

\begin{document}

\title{Information Spectrum Approach to the Source Channel Separation Theorem}
\author{
Nir~Elkayam,~\IEEEmembership{Student Member,~IEEE,}
Meir~Feder,~\IEEEmembership{Fellow,~IEEE}
}
\maketitle

\subsection*{\centering Abstract}
%\textit, \textbf
\textit{A source-channel separation theorem for a general channel has recently been shown by Aggrawal et. al \cite{agarwal2013universal}.
This theorem states that
%In this paper we provide simple proofs of the statements given in \cite{agarwal2013universal} which essentially say, that
if there exist a coding scheme that achieves a maximum distortion level $d_{max}$ over a general channel \bW, then reliable communication can be accomplished over this channel at rates less then $R(d_{max})$, where $R(\cdot)$ is the rate distortion function of the source. The source, however, is essentially constrained to be discrete and memoryless (DMS).
In this work we prove a stronger claim where the source is general, satisfying only a ``sphere packing optimality'' feature, and the channel is completely general.
Furthermore, we show that if the channel satisfies the strong converse property as define by Han \& \verdu\ \cite{DBLP:journals/tit/VerduH94}, then the same statement can be made with $d_{avg}$, the average distortion level, replacing $d_{max}$. Unlike the proofs in \cite{agarwal2013universal}, we use information spectrum methods to prove the statements and the results can be quite easily extended to other situations. }

\section{Introduction}\label{sec:intro}

The source channel separation theorem, first proved by Shannon \cite{shannon1959coding} for the transmission of discrete memoryless source (DMS) over discrete memoryless channel (DMC) states that the \emph{separation strategy} is optimal. This means optimal performance can be attained by first compressing the source output to the desired distortion level and then reliably communicating the compressed bits over the channel. The separation theorem was later extended to indecomposable channels \cite{gallager1968information}.
%\textbf{NEW: }
The almost lossless case (transmission codes) was handled at \cite{vembu1995source} and a general condition is given for the separation theorem to hold.

%The source channel separation theorem, first proved by Shannon \cite{shannon1959coding} for the transmission of discrete memoryless source (DMS) over discrete memoryless channel (DMC) states that an optimal strategy for transmission of information in a lossy manner over a noisy channel can be \emph{separation}. First, compress the source output to the desired distortion level and then reliable communication of the compressed bits over the channel. This was extended to indecomposable channels \cite{gallager1968information}.

Joint Source-Channel Coding (JSCC) refers to the case where such a separation is not used. In some cases, e.g. binary sources over BSC or Gaussian source over AWGN channel, separation can attain the optimal performance yet a simpler JSCC strategy can be used, \ie, uncoded transmission \cite{gastpar2003code}. In some cases separation is suboptimal. A simple example would be the symmetric binary source and a compound memoryless BSC where the flipping probability is drawn ahead of the block and stays fixed for the whole block. In this case uncoded transmission is optimal and it is strictly suboptimal to use separation.

Some cation is needed here, because there are two senses of optimality when a distortion measure is given - the \textbf{maximal distortion} level and \textbf{average distortion} level. The separation relative to the maximal distortion level is easier to accomplish, as this allows us to increase the distortion as long as we do not exceed the desired distortion level. So even if the average distortion in the original scheme is much less then the maximum distortion, the separation strategy yield the desired maximum distortion but the average distortion might be increased.

When average distortion level is used, we should follow the specific distortion, which can be large or small as the channel condition are good or bad. This is much like the variable rate channel capacity \cite{DBLP:journals/tit/VerduS10} which tries to capture the whole spectrum of channel conditions - when the channel provides good conditions %(without feedback),
lower distortion level can be achieved or more bits can be reliable transmitted over the channel. A separation strategy in that case will use successive refinement of the source \cite{equitz1991successive}, \cite{rimoldi1994successive} and then transmission of the bits over a channel with variable rate channel capacity so that the better the channel, the lower the distortion \cite{tian2008successive} can be made. In many cases this separation strategy is suboptimal.

In transmitting a DMS over a DMC the cases of average distortion level and maximal distortion level coincide \cite{YuvalKochmanConverse}. However, for other channels this is not always the case even for DMS's. We will see that for indecomposable channels, this is always true.

Information spectrum methods \cite{DBLP:journals/tit/VerduH94}, \cite{koga2002information} provide a very simple formalization and intuition into channel capacity in almost every communication situation including unicast, multiple access, broadcast and other situations \cite{somekh2006general}. The problem is that the expressions given with these methods are usually not useful when we want to compute the channel capacity. Nevertheless, we use information spectrum methods in this work and attain the desired results.

In this paper we deal with sources for which the sphere packing bound is tight. This means that the sphere of radius $d$ around any reconstruction point contains at most $2^{-nR(d)}$ of the typical space. For these source we first prove a generalization of the Han-Verd{\'u} converse lemma \cite{DBLP:journals/tit/VerduH94} that connects the \textit{rate distortion spectrum} (that is, the probability that the random variable $D$, the instantaneous distortion level exceed some level $d$) to the information spectrum of the channel. Then we use this to prove a very general separation theorem for the case where maximum distortion level criteria is given. For channels that satisfy the strong converse, which includes DMC and ergodic channels, we prove a separation theorem for the average distortion level as well. Actually, for these kind of channels we prove that the average and maximal distortion level coincide.

\section{Preliminaries}
\subsection{Notation}

This paper uses lower case letters (e.g. $x$) to denote a particular value of the corresponding random variable denoted in capital letters (e.g. $X$). Calligraphic fonts (e.g. \calX) represent a set.

The random variable $D$ represents the instantaneous distortion level. $\ED=d_{avg}$ is the average distortion level, and $d_{max}$ will be used to denote the maximum distortion level.

We will use the $o(\cdot)$ notation to denote terms that goes to $0$ w.r.t the argument. Mainly, $o(n)$ will denote a sequence $\epsilon_n$ such that $ \limtoinf \epsilon_n/n=0 $ and $o(1)$ denote a sequence that converges to 0. Throughout this paper $\log$  will be defined to the base 2 unless otherwise indicated. $\PR{A}$ will denote the probability of the event $A$.

\subsection{Definition}

\begin{definition}[\cite{koga2002information} Source and Reproduction alphabet, Distortion measure, Rate distortion function]
A general source \bS\ is defined as an infinite sequence of random variable on $\cS_n$. The reproduction alphabet is defined over the set $\chS_n$. A distortion measure is a function $ d_n: \cS_n \times \chS_n \rightarrow \mathbb{R_{+}}$.
%where $\chS_n$ the reproduction alphabet.
There are several rate distortion functions that can be defined according to the different performance requirements (see \cite{koga2002information} ch.5.3).
\begin{itemize}
  \item $R_{fm}(d)$ - fixed length code, maximum distortion criterion.
  \item $R_{fa}(d)$ - fixed length code, average distortion criterion.
  \item $R_{vm}(d)$ - variable length code, maximum distortion criterion.
  \item $R_{va}(d)$ - variable length code, average distortion criterion.
\end{itemize}
\end{definition}

\begin{definition}[Sphere packing optimal Source (SPO)]
A general source \bS\ is said to be \textbf{Sphere packing optimal} if there exist a subsets $\cA_n \subset \cS_n$ and $k_n$, such that:
\begin{align} \label{def:spo_source}
&\limtoinf \PR{\Selmnt \notin \cA_n} = 0 \\
&\PR{\Selmnt \in \cA_n, d(\Selmnt, \Shelmnt) \leq d} \leq 2^{-n\BRA{R(d)+k_n}} \\
&\limtoinf k_n = 0
\end{align}
for each $d \geq 0$ and $ \Shelmnt \in \chS_n $.
\end{definition}

\begin{remark} In the appendix we demonstrate that DMS's are SPO. The set $\cA_n$ will be the strong typical sequences relative to the source distribution. The Gaussian source with mean square distortion is also SPO as can be shown by a straight forward calculation, given in the appendix as well.
\end{remark}

\begin{remark} For SPO sources that have a reference word (see \cite{koga2002information} Theorem 5.3.1), the following different notions of rate functions are equal $$R_{fm}(d) = R_{fa}(d) = R_{vm}(d).$$ The common value will be denoted $R(d)$ without the subscript. This is also shown in the appendix.
\end{remark}

\begin{definition}[General Channel] A general channel is a sequence of transition matrices $ \bW = \BRAs{ W_n : \cX_n \rightarrow \cY_n } $ where $ W^n(y|x) $ denotes the conditional probability of $y$ given $x$. Throughout this paper we will assume that the channel has finite input and output space. \footnote{ Extension of the results to abstract input and output spaces $\BRA{\cX_n, \cY_n}$ require subtle handling, see \cite{koga2002information}, but is possible to whenever situation which information spectrum can be used. }
\end{definition}

\begin{remark}
The source and channel input and output will be written as $\cS_n, \cX_n, \cY_n, \chS_n$ to indicate that it is usually not the $n^{th}$ order cartesian product. As an example let $\cS_n=\cS^{nR}, \chS_n=\chS^{nR}$ and $\cX_n=\cX^n, \cY_n=\cY^n$. This allows us to simplify the notation and to avoid unnecessary "rate" indication. We will use $x$ and $y$ to denote the input and output. Occasionally, we will omit the subscript $n$ and assume that it is understood from the context.
\end{remark}
\begin{remark}
Throughout the paper we assume that the sets $\BRA{\cS_n, \cX_n, \cY_n, \chS_n}$ are discrete. However, the results are valid for any abstract spaces for which information spectrum method can be applied see \cite{koga2002information}.
\end{remark}

\begin{definition}[JSCC Scheme] A general JSCC Scheme of the source $\bS$ over the channel $\bW$, includes:
\begin{itemize}
  \item Encoding function: $E_n: \cS_n \rightarrow \cX_n$
  \item Decoding function: $D_n: \cY_n \rightarrow \chS_n $.
\end{itemize}
There are several random variables defined:
\begin{itemize}
  \item $\bX_n = E_n(\bS_n)$ - The random variable of the channel input.
  \item $\bY_n$ - The output from the channel $\bW$ resulted from the input $\bX_n$.
  \item $D = d(\bS_n, D_n(\bY_n))$ - The instantaneous distortion.
\end{itemize}

%The R.V. $D$ will be the resulted distortion from the concatenation of the source, encoding, channel, and decoding:  $ \cS_n \overset{E_n}{\rightarrow} \cX_n \overset{W_n}{\rightarrow} \cY_n \overset{D_n}{\rightarrow} \chS_n $, that is, $D=d(\cS_n, \chS_n)$.
Sometimes we will write $x(\Selmnt)$ for $E_n(\Selmnt)$. The maximum distortion level $d_{max}$ is the infimum of the set of numbers $\alpha$, such that $\PR{D \geq \alpha} \toinf 0$. The average distortion level is $\ED=d_{avg}$.
\end{definition}

\subsection{Information Spectrum Notation}

\begin{definition}[Liminf in Probability]
If $\cA_n$ is a sequence of random variables, its \textbf{liminf in probability} is the supremum of all the reals $\alpha$ for
which $\PR{\cA_n \geq \alpha} \toinf 0$.
\end{definition}

We will use $i(a;b)$ to denote the information density between two outcomes from correlated random variables $A$ and $B$. Specifically, $i(a;b) = \frac{1}{n}\log\BRA{\frac{p(a|b)}{p(a)}} $. We will omit the $n$ and the indication to which random variable produce the outcomes as it will be understood from the context.

\begin{definition}[information density]
Given random variables $\bX_n, \bY_n$ with joint distribution $p(x,y)$ on $\cX_n \times \cY_n$, the \textbf{information density} is the function defined on $\cX_n \times \cY_n$:
$$ i_{\bX_n, \bY_n}(x;y) = \frac{1}{n}\log\BRA{\frac{p_{\bY_n|\bX_n}(y|x)}{p_{\bY_n}(y)}} $$
\end{definition}

\begin{remark} we will write $i\BRA{x,y}$ instead of $i_{\bX_n, \bY_n}(x;y)$ to avoid cumbersome notation, as it will be understood which random variables is in use.
\end{remark}

\begin{definition}[Inf-Information Rate]
The liminf in probability of the sequence of random variables $\frac{1}{n}i\BRA{\bX_n, \bY_n}$ will be referred to
as the the inf-information rate of the pair $ \bX_n, \bY_n $ and will be denoted as $ \underline{I}\BRA{\bX_n; \bY_n}$.
\end{definition}

In \cite{DBLP:journals/tit/VerduH94} it is shown that the capacity of a general channel is given by:
$$ C = \sup_{\bX_n}\underline{I}\BRA{\bX_n; \bY_n} $$ where the supremum is over the prior distribution $\bX_n$.

\begin{definition}[Strong converse]
The epsilon-capacity $C_{\epsilon}$ is the supremum of all rates $R$, for which there exist a coding scheme with rate $R$ and error probability less then $ \epsilon$. A channel is said to satisfy the strong converse property, if $C = \lim_{\epsilon \to 0 } C_{\epsilon}$. This means that is there exist a coding scheme of rate $R$ with error less then 1, then there exist a coding scheme of rate $R$ with error probability approach 0.
\end{definition}

%\subsection{Technical lemma - where to put it?}
Throughout the sequel we will need to use the following lemma which say that if there exist a probability mass to the right of the mean of a non-negative random variable then there must be a probability mass left to it to.
\begin{lemma} \label{lem:pos_D_mean}
Let $D$ be a non-negative random variable with $\mu = \ED < \infty$. If there exist $d_1>0$ and $\epsilon_1>0$ such that: $\PR{D > \mu+d_1} > \epsilon_1$, then there exist $d_2>0$ and $\epsilon_2>0$ such that: $\PR{D < \mu-d_2} > \epsilon_2$.
\end{lemma}

\begin{proof}
The proof is given in the appendix \ref{App:AppendixB}.
\end{proof}

\section{Joint Source Channel lemma (Unicast)}\label{sec:lemmas}
In this section we state and prove a generalization of the Han -\verdu converse lemma \cite{DBLP:journals/tit/VerduH94} which relates the error rate probability to the information spectrum. Here we can connect the instantaneous distortion level $D$ with the information spectrum of the channel.
%\textbf{NEW: } 
A similar, but different converse result was given in \cite[Theorem 1]{DBLP:journals/tit/KostinaV13} for a general source. We, however, provide the following lemma that holds for an SPO source which is enough for our purposes.
Let:
$$ L_{R(d)}^{\gamma} = \BRAs{ \BRA{\Selmnt,y} : i \BRA{x(\Selmnt);y} \leq R(d)-\gamma} $$
This is exactly the set which is used in the definition of channel capacity in terms of information spectrum.

\begin{lemma}[JSCC Converse lemma]\label{lem:JSC_Converse}
If the source \bS\ is SPO, then:
$$ \PR{D > d} \geq \PR{L_{R(d)}^{\gamma}} -2^{-n \BRA{\gamma+k_n}} - \PR{\Selmnt \notin \cA_n} $$
\end{lemma}

\begin{proof} %(Sketch)
The term $\PR{L_{R(d)}^{\gamma}}$ can be bounded by:
\begin{align} \label{eq:con_uni_1}
&\PR{L_{R(d)}^{\gamma}} \leq \PR{D > d} + \PR{\Selmnt \notin \cA_n} + \PR{L_{R(d)}^{\gamma} \cap \BRA{D \leq d}\cap\BRA{\Selmnt \in \cA_n}}
\end{align}
Continues with the last term in \eqref{eq:con_uni_1}:
\begin{align} \label{eq:con_uni_2}
&\PR{L_{R(d)}^{\gamma} \cap \BRA{D \leq d}\cap\BRA{\Selmnt \in \cA_n}} \notag \\
& = \sum_{\substack{\Selmnt \in \cA_n\\ \BRA{x(\Selmnt),y} \in L_{R(d)}^{\gamma}\\ D \leq d}}p(\Selmnt)p(y|x(\Selmnt)) \notag  \\
&\overset{(a)}{\leq} \sum_{\substack{\Selmnt \in \cA_n, y \in \cY_n \\ D \leq d}}p(\Selmnt)p(y) \cdot  2^{nR(d)-n \gamma} \\
&= \sum_{y \in \cY_n }p(y) \cdot  2^{nR(d)-n \gamma} \cdot \sum_{\substack{\Selmnt \in \cA_n\\ \hat{s}^n=D_n(y), d(\Selmnt,\hat{s}^n) \leq d}}p(\Selmnt) \notag \\
&\overset{(b)}{\leq} \sum_{y \in \cY_n }p(y) \cdot  2^{nR(d)-n \gamma} \cdot 2^{-n\BRA{R(d)+k_n}} \notag \\
&= 2^{-n\BRA{\gamma+k_n}} \notag
\end{align}

where $(a)$ follows because $\BRA{\Selmnt,y} \in L_{R(d)}^{\gamma} \Rightarrow p(y|x(\Selmnt)) \leq p(y) \cdot 2^{nR(d)-n\gamma}$ and $(b)$ follows from the SPO assumption of the source \bS.
Combining \eqref{eq:con_uni_1} and \eqref{eq:con_uni_2} complete the proof of the inequality.
\end{proof}
\section{Source Channel Separation Theorems (Unicast)}\label{sec:uni_thrm}
\subsection{Maximum distortion }
\begin{theorem}[Source channel separation, maximum distortion] \label{sec:uni_thrm1} If the maximum distortion in the joint source channel coding of an SPO source \bS\ is $d_{max}$ then any rate less then $R(d_{max})$ is achievable, \ie\ the capacity of the channel \bW\ is at least $R(d_{max})$.
\end{theorem}
\begin{proof}
Fix a sequence $\gamma_n \toinf 0$ such that $n\BRA{\gamma_n+k_n} \toinf \infty$. Let $X_n=E_n(S_n)$ and $Y_n$ the output of $X_n$ through the channel \bW.
From lemma \ref{lem:JSC_Converse} we have:
$$ \PR{D > d_{max}} \geq \PR{L_{R(d_{max})}^{\gamma_n}}-2^{-n \BRA{\gamma_n+k_n}}+o(1) $$
Since $ \limtoinf \PR{D > d_{max}} = 0 $ and $2^{-n\BRA{\gamma_n+k_n}} \toinf 0 $ we get $ \limtoinf \PR{L_{R(d_{max})}^{\gamma_n}} = 0 $. Since $\gamma_n \toinf 0$ this implies that:
$$ \limtoinf \PR{i \BRA{x(\Selmnt);y}\leq R(d)} = 0 $$
which proves that $ \underline{I}\BRA{\bX_n; \bY_n} \geq R(d)$ so that channel capacity is at least $R(d)$ and separation is possible.
\end{proof}

\subsection{Average distortion }
\begin{theorem}[Joint Source channel Separation, average distortion] \label{sec:uni_thrm2}
For an optimal JSCC scheme, of an SPO source over a channel that satisfy the strong converse property, the notion of average distortion and maximum distortion coincide.
%If the average distortion in the joint source channel coding of an SPO source \bS\ is $d_{avg}$ and the source is SPO, then any rate less then $R(d_{avg})$ is achievable, \ie\ the capacity of the channel \bW\ is at least $R(d_{avg})$.
\end{theorem}
\begin{proof}
Suppose not. Then there exist $\tau_1$ and $ \epsilon_1 $ such that: $$\limsuptoinf \PR{D > \ED+\tau_1} \geq \epsilon_1$$
Fix a sequence $\gamma_n \toinf 0$ such that $n\BRA{\gamma_n+k_n} \toinf \infty$. By lemma \ref{lem:pos_D_mean} we have: $\liminftoinf \PR{D > \ED-\tau} > 1-\epsilon$ for some $\tau>0$ and $\epsilon>0$).
Now:
\begin{align} \label{eq:uni_thm_1}
1-\epsilon &\geq \liminftoinf \PR{D > \ED-\tau} \notag \\
&\geq \liminftoinf \BRA{\PR{L_{R(\ED-\tau)}^{\gamma_n}}-2^{-n\BRA{\gamma_n+k_n}}+o(1) }
\end{align}
For $n$ large enough we have that $\liminftoinf \PR{L_{R(\ED-\tau)}^{\gamma_n}}$ bound away from 1, so we can transmit at rate $R(\ED-\tau)-\gamma_n$ with error probability less then 1. Since the channel satisfy the strong converse, this means that we can transmit at rate $R(\ED-\tau)$ with error approach 0, so the JSCC is not optimal because with separation we can transmit at maximum distortion level $\ED-\tau$ which is better the average $\ED$. (We used here the fact that $R_{fm}(d) = R_{fa}(d)$).
\end{proof}

For the average distortion rate we can integrate the bound in lemma \ref{lem:JSC_Converse} to get:
\begin{theorem}[General lower bound on JSCC distortion] \label{sec:uni_thrm2}
If there is a bound on the maximum distortion $d_{max}$ then:
$$ \ED \geq \int_{0}^{d_{max}}\PR{L_{R(\alpha)}^{\gamma}}d\alpha -d_{max} \cdot 2^{-n \gamma+o(n)} - d_{max} \cdot \PR{\Selmnt \notin \cA_n} $$
\end{theorem}
\begin{proof}
Use $\ED = \int_{0}^{d_{max}}\PR{D>\alpha}d\alpha $ and lemma \ref{lem:JSC_Converse}.
\end{proof}
\section{The Multiple Access Case}
In this section we demonstrate in a loose manner how the same method can applied to more general communication situations. Specifically, a 2-users multiple access channel. Let $\bW$ now represent a multiple access channel with 2 inputs -  $ \bW = \BRAs{ W_n : \cX_{1,n} \times \cX_{2,n}\ \rightarrow \cY_n } $ and the sources $\bS_1, \bS_2$ are two uncorrelated SPO sources.
Let:
$$ L_{R_1(d)}^{\gamma} = \BRAs{ \BRA{\Selmnto,\Selmntt,y} : i \BRA{x_1(\Selmnto);y|x_2(\Selmnto)} \leq R_1(d)-\gamma}$$
$$ L_{R_2(d)}^{\gamma} = \BRAs{ \BRA{\Selmnto,\Selmntt,y} : i \BRA{x_2(\Selmntt);y|x_1(\Selmnto)} \leq R_2(d)-\gamma}$$
$$ L_3^{\gamma}        = \BRAs{ \BRA{\Selmnto,\Selmntt,y} : i \BRA{x_1(\Selmnto), x_2(\Selmntt);y} \leq R_1(d)+R_2(d)-\gamma}$$
Also, let $T=L_{R_1(d)}^{\gamma}\cup L_{R_2(d)}^{\gamma} \cup L_3^{\gamma}$
This is exactly the event which is used in the definition of channel capacity in term of information spectrum.

\begin{lemma}[JSCC Converse lemma, Multiple Access ]\label{lem:JSC_Converse_MA}
If the source $\bS_1, \bS_2$ are SPO, then:
$$ \PR{D_1 > d_1 \cup D_2 > d_2} \geq \PR{T} -2^{-n \gamma+o(n)} - \PR{\Selmnto \notin A_{n1}} - \PR{\Selmntt \notin A_{n2}}$$
\end{lemma}
\begin{proof} %(Sketch)
The term $\PR{T}$ can be bounded by:
\begin{align} \label{eq:conv_multi_1}
&\PR{T} \leq \PR{D_1 > d_1 \cup D_2 > d_2} \notag \\
&+\PR{\Selmnto \notin A_{n1}} + \PR{\Selmntt \notin A_{n2}} \\
&+ \PR{T \cap \BRA{D_1 \leq d_1} \cap \BRA{D_2 \leq d_2}\cap\BRA{\Selmnto \in A_{n1}} \cap\BRA{\Selmntt \in A_{n2}} } \notag
\end{align}
Using the union bound on the $T$ term we get 3 terms which can be bounded like before. We'll demonstrate for the term which contains $ L_{R_1(d)}^{\gamma}$, the others follow the same lines.
\begin{align} \label{eq:conv_multi_2}
&\PR{L_{R_1(d)}^{\gamma} \cap \BRA{D_1 \leq d_1} \cap \BRA{D_2 \leq d_2}\cap\BRA{\Selmnto \in A_{n1}} \cap\BRA{\Selmntt \in A_{n2}} } \notag \\
&\leq \PR{L_{R_1(d)}^{\gamma} \cap \BRA{D_1 \leq d_1} \cap\BRA{\Selmnto \in A_{n1}} } \notag \\
& = \sum_{...}p(\Selmnto)p(\Selmntt)p(y|x(\Selmnto), x(\Selmntt)) \notag \\
& \overset{(a)}{\leq} \sum_{...}p(\Selmnto)p(\Selmntt)p(y|x(\Selmntt)) \cdot  2^{nR_1(d)-n \gamma}  \notag \\
& = \sum_{...}p(\Selmntt)p(y|x(\Selmntt)) \cdot  2^{nR_1(d)-n \gamma} \sum_{...}p(\Selmnto) \notag \\
& \overset{(b)}{\leq} \sum_{...}p(\Selmntt)p(y|x(\Selmntt)) \cdot  2^{nR_1(d)-n \gamma} \cdot 2^{-nR_1(d)+o(n)}  \notag \\
& = 2^{-n\gamma+o(n) }
\end{align}

where $(a)$ follows because $p(y|x(\Selmnto), x(\Selmntt)) \leq p(y|x(\Selmntt)) \cdot  2^{nR_1(d)-n \gamma} $ for $\BRA{\Selmnto, \Selmntt,y} \in L_{R_1(d)}^{\gamma}$ and $(b)$ follows from the SPO assumption of the source $\bS_1$.
Combining \eqref{eq:conv_multi_1} and \eqref{eq:conv_multi_2} complete the proof of the inequality.
\end{proof}

We can use this lemma to prove 2 JSCC Seperation analog to \ref{sec:uni_thrm1} and \ref{sec:uni_thrm2}.

\section{Further Research}
There are several ways to continue with this research.
\begin{itemize}
  \item SPO sources - Find out which sources are SPO's. Furthermore, examine whether the SPO condition can be relaxed. 
  % One example for such relaxation, is giving up the requirement that for all $\Shelmnt$, the second condition must be met.
  \item Another interesting question is whether sources for which the average distortion rate equals the maximum distortion rate are SPO's. If this result holds, it will provide a source channel separation for ergodic stationary sources with sub-additive distortion measure.
  \item Examine the average distortion case for source transmission over channels that do not satisfy the strong converse.
  \item Network case: Examine the correlated source \cite{tian2010optimality}.
\end{itemize}
These problems and probably additional questions are left for further research.
\appendices

\section{SPO Sources} \label{App:AppendixA}
\subsection{Gaussian Source with mean square distortion}
Here we provide sketch to prove that gaussian source with mean square distortion is SPO. Let $\bS$ be a gaussian source with variable $\sigma^2$.
Let $ \varepsilon_n \geq 0$ be a sequence such that $ \limtoinf \varepsilon_n = 0$ and $ \limtoinf n \cdot \varepsilon_n^2 = \infty $. $\cA_n = \BRAs{X^n: \left| \frac{1}{n}\sum_{i=1}^n X_i^2 - \sigma^2 \right| \leq \varepsilon_n} $. By the Chebyshev's inequality we have $\PR{\cA_n} \toinf 1$. For any reproduction vector $y^n=\BRA{y_i}$ define $ d = \frac{1}{n}\sum_{i=1}^{n} (X_i -y_i )^2 $. We need to prove that:
$$\PR{\BRA{d \leq D} \cap \cA_n } \leq 2^{-nR(D)+o(n)} = 2^{-\frac{n}{2}\log\BRA{\frac{\sigma^2}{D}}+o(n)} $$

\begin{align} \label{eq:gauss_spo1}
&\PR{\BRA{d \leq D} \cap \cA_n } \notag \\
& = \int_{X^n \in \cA_n, d \leq D}\BRA{2\pi\sigma^2}^{-\frac{n}{2}}e^{-\frac{\sum_{i=1}^n X_i^2}{2\sigma^2}} dX \notag \\
& \leq \int_{X^n \in \cA_n, d \leq D}\BRA{2\pi\sigma^2}^{-\frac{n}{2}}e^{-n\frac{\sigma^2 - \varepsilon_n}{2\sigma^2}} dX \notag \\
& \leq \BRA{2\pi\sigma^2}^{-\frac{n}{2}}e^{-n\frac{\sigma^2 - \varepsilon_n}{2\sigma^2}} \int_{d \leq D} dX \notag \\
& \overset{(a)}{=} \BRA{2\pi\sigma^2}^{-\frac{n}{2}} e^{-\frac{n}{2}+n\frac{\varepsilon_n}{2\sigma^2}} \frac{\pi^{\frac{n}{2}}}{\Gamma \BRA{\frac{n}{2}+1}} \BRA{\sqrt{nD}}^{n} \notag \\
& \overset{(b)}{=} \BRA{2\sigma^2}^{-\frac{n}{2}} e^{-\frac{n}{2}+n\frac{\varepsilon_n}{2\sigma^2}} \frac{1}{\BRA{\frac{n}{2}}!} \BRA{nD}^{\frac{n}{2}} \notag \\
& \overset{(c)}{\approx} \BRA{2\sigma^2}^{-\frac{n}{2}} e^{-\frac{n}{2}+n\frac{\varepsilon_n}{2\sigma^2}} \frac{1}{\BRA{\frac{n}{2e}}^{\frac{n}{2}}} \BRA{nD}^{\frac{n}{2}} \notag \\
& = \BRA{\sigma^2}^{-\frac{n}{2}}e^{n\frac{\varepsilon_n}{2\sigma^2}} \BRA{D}^{\frac{n}{2}} \notag \\
& = 2^{-\frac{n}{2} \log \BRA{\frac{\sigma^2}{D}}} e^{n\frac{\varepsilon_n}{2\sigma^2}} \notag
\end{align}
$(a)$ follows from the well known formula for the volume of $n$-dimensional sphere with radius $\sqrt{nD}$. \\
$(b)$ follows because we canceled $\pi$, and assuming $n$ is even for which value the formula for $\Gamma$ because simpler. \\
$(c)$ is the stirling's approximation where we ignored the $\sqrt{2 \pi n}$ term which is not contributing to the exponent.
%As we can choose arbitrary small $\varepsilon$ this complete the proof.

\subsection{DMC with finite distortion}
A proof that uses the method of type can be given along the lines of lemma 1 in \cite{lomnitz2011communication}.
We'll just note that we can control the type's of $x$ by the intersection with $\cA_n$.
\subsection{Notions of distortions}

\begin{remark} For a source that satisfy the SPO property we can show that the average $R_{fa}(d)$ and the maximum rate function $R_{fm}(d)$ coincide.
\end{remark}

To see this, let $Y$ be the result of encoding and decoding with average distortion $d$ and assume that $R_{fa}(d) < R_{fm}(d)$. Then we have:
\begin{align*} %\label{def:spo_avg_max}
\PR{ \BRA{d\BRA{X,Y} \leq \alpha} \cap \cA_n} &\leq \sum_{i=1}^{2^{nR_{fa}(d)}} \PR{ \BRA{d\BRA{X,y_i} \leq \alpha} \cap \cA_n} \\
&\leq 2^{nR_{fa}(d)} \cdot 2^{-n\BRA{R_{fm}(\alpha)+k_n}} \\
&\leq 2^{n\BRA{R_{fa}(d)-R_{fm}(d)+k_n}} \toinf 0
\end{align*}

Continue:
\begin{align*} %\label{def:spo_avg_max1}
& 1-\PR{\BRA{d\BRA{X,Y} \leq \alpha} \cap \cA_n} \\ % &= \PR{\BRA{\BRA{d\BRA{X,Y} \leq \alpha} \cap \cA_n}^c} \\
&= \PR{ \BRA{d\BRA{X,Y} > \alpha} \cup \cA_n^c } \\
& \leq \PR{ d\BRA{X,Y} > \alpha } + \PR{ \cA_n^c }
\end{align*}
Rearranging we get:
\begin{align*} %\label{def:spo_avg_max7}
\PR{ d\BRA{X,Y} > \alpha } \geq 1-\PR{\BRA{d\BRA{X,Y} \leq \alpha} \cap \cA_n}-\PR{ \cA_n^c }
\end{align*}

From:
\begin{align*}
\E{d\BRA{X,Y}} &= \int_{0}^{\infty}\PR{ d\BRA{X,Y}  > \alpha } d\alpha \\
& \geq \int_{0}^{d}\PR{ d\BRA{X,Y}  > \alpha } d\alpha \\
& \geq \int_{0}^{d}\BRA{1-\PR{\BRA{d\BRA{X,Y} \leq \alpha} \cap \cA_n}-\PR{ \cA_n^c }} d\alpha \\
& = d - \int_{0}^{d}\PR{\BRA{d\BRA{X,Y} \leq \alpha} \cap \cA_n} d\alpha - d \cdot \PR{ \cA_n^c }\\
\end{align*}
The results then follows because $\PR{ \cA_n^c } \toinf 0$ and $\PR{\BRA{d\BRA{X,Y} \leq \alpha} \cap \cA_n} \toinf 0$ for $ \alpha < d $.

\begin{remark} For a source that satisfy the SPO property we can show that the average $R_{fm}(d)$ and the maximum rate function $R_{vm}(d)$ coincide.
\end{remark}
To see this, let $l_i$ be the length of the $i^{th}$ code word ordered according to their probability. By using kraft's inequality it can be shown that $R = \frac{1}{n} E \BRA{l_i} \geq \frac{1}{n}H(l_i)$, \eg\ \cite[Theorem 5.6.1]{koga2002information}. Since $\PR{l_i} \leq 2^{-nR_{fm}(d)}$ it follows that $H(l_i) \geq nR_{fm}(d)$.

\section{Proof of lemma \ref{lem:pos_D_mean}} \label{App:AppendixB}

\begin{proof}
For a non-negative random variable we have: $ \mu = \int_{0}^{\infty}\PR{D>x}dx $. Now:
\begin{align*}
\int_{\mu}^{\infty}\PR{D>x}dx
&\geq \int_{\mu}^{\mu+d_1}\PR{D>x}dx \\
&\geq \int_{\mu}^{\mu+d_1}\epsilon_1dx = d_1 \epsilon_1
\end{align*}
So we have: $ \int_{0}^{\mu}\PR{D>x}dx < \mu - d_1 \epsilon_1$. Let $d_2$ and $\epsilon_2$ be such that $(\mu-d_2)(1-\epsilon_2) > \mu-d_1 \epsilon_1$. We must have:  $\PR{D < \mu-d_2} > \epsilon_2$ because otherwise: %The same argument show that for $d_2 = \frac{1}{\mu - d_1 \epsilon_1}$ we must have
\begin{align}
\mu - d_1 \epsilon_1 & > \int_{0}^{\mu}\PR{D>x}dx \notag \\
&> \int_{0}^{\mu-d_2}\PR{D>x}dx \\
&> \int_{0}^{\mu-d_2}\BRA{1-\epsilon_2}dx \\
& = \BRA{\mu-d_2}\BRA{1-\epsilon_2} \notag
\end{align}
\end{proof}

\bibliographystyle{IEEEtran}
\bibliography{bibfile}

\begin{thebibliography}{10}
\providecommand{\url}[1]{#1}
\csname url@samestyle\endcsname
\providecommand{\newblock}{\relax}
\providecommand{\bibinfo}[2]{#2}
\providecommand{\BIBentrySTDinterwordspacing}{\spaceskip=0pt\relax}
\providecommand{\BIBentryALTinterwordstretchfactor}{4}
\providecommand{\BIBentryALTinterwordspacing}{\spaceskip=\fontdimen2\font plus
\BIBentryALTinterwordstretchfactor\fontdimen3\font minus
  \fontdimen4\font\relax}
\providecommand{\BIBforeignlanguage}[2]{{%
\expandafter\ifx\csname l@#1\endcsname\relax
\typeout{** WARNING: IEEEtran.bst: No hyphenation pattern has been}%
\typeout{** loaded for the language `#1'. Using the pattern for}%
\typeout{** the default language instead.}%
\else
\language=\csname l@#1\endcsname
\fi
#2}}
\providecommand{\BIBdecl}{\relax}
\BIBdecl

\bibitem{agarwal2013universal}
M.~Agarwal, S.~Mitter, and A.~Sahai, ``A universal, operational theory of
  unicast multi-user communication with fidelity criteria,'' \emph{arXiv
  preprint arXiv:1302.5860}, 2013.

\bibitem{DBLP:journals/tit/VerduH94}
S.~Verd{\'u} and T.~S. Han, ``A general formula for channel capacity,''
  \emph{IEEE Transactions on Information Theory}, vol.~40, no.~4, pp.
  1147--1157, 1994.

\bibitem{shannon1959coding}
C.~E. Shannon, ``Coding theorems for a discrete source with a fidelity
  criterion,'' \emph{IRE Nat. Conv. Rec}, vol.~4, no. 142-163, 1959.

\bibitem{gallager1968information}
\BIBentryALTinterwordspacing
R.~Gallager, \emph{Information theory and reliable communication}.\hskip 1em
  plus 0.5em minus 0.4em\relax Wiley, 1968. [Online]. Available:
  \url{http://books.google.co.il/books?id=Uc3uAAAAMAAJ}
\BIBentrySTDinterwordspacing

\bibitem{vembu1995source}
S.~Vembu, S.~Verdu, and Y.~Steinberg, ``The source-channel separation theorem
  revisited,'' \emph{Information Theory, IEEE Transactions on}, vol.~41, no.~1,
  pp. 44--54, 1995.

\bibitem{gastpar2003code}
M.~Gastpar, B.~Rimoldi, and M.~Vetterli, ``To code, or not to code: Lossy
  source-channel communication revisited,'' \emph{Information Theory, IEEE
  Transactions on}, vol.~49, no.~5, pp. 1147--1158, 2003.

\bibitem{DBLP:journals/tit/VerduS10}
S.~Verd{\'u} and S.~Shamai, ``Variable-rate channel capacity,'' \emph{IEEE
  Transactions on Information Theory}, vol.~56, no.~6, pp. 2651--2667, 2010.

\bibitem{equitz1991successive}
W.~H.~R. Equitz and T.~M. Cover, ``Successive refinement of information,''
  \emph{Information Theory, IEEE Transactions on}, vol.~37, no.~2, pp.
  269--275, 1991.

\bibitem{rimoldi1994successive}
B.~Rimoldi, ``Successive refinement of information: Characterization of the
  achievable rates,'' \emph{Information Theory, IEEE Transactions on}, vol.~40,
  no.~1, pp. 253--259, 1994.

\bibitem{tian2008successive}
C.~Tian, A.~Steiner, S.~Shamai, and S.~N. Diggavi, ``Successive refinement via
  broadcast: Optimizing expected distortion of a gaussian source over a
  gaussian fading channel,'' \emph{Information Theory, IEEE Transactions on},
  vol.~54, no.~7, pp. 2903--2918, 2008.

\bibitem{YuvalKochmanConverse}
D.~Wang, A.~Ingber, and Y.~Kochman, ``A strong converse for joint
  source-channel coding,'' in \emph{Information Theory Proceedings (ISIT), 2012
  IEEE International Symposium on}.\hskip 1em plus 0.5em minus 0.4em\relax
  IEEE, 2012, pp. 2117--2121.

\bibitem{koga2002information}
H.~Koga \emph{et~al.}, \emph{Information-spectrum methods in information
  theory}.\hskip 1em plus 0.5em minus 0.4em\relax Springer, 2002, vol.~50.

\bibitem{somekh2006general}
A.~Somekh-Baruch and S.~Verd{\'u}, ``General relayless networks: representation
  of the capacity region,'' in \emph{Information Theory, 2006 IEEE
  International Symposium on}.\hskip 1em plus 0.5em minus 0.4em\relax IEEE,
  2006, pp. 2408--2412.

\bibitem{DBLP:journals/tit/KostinaV13}
V.~Kostina and S.~Verd{\'u}, ``Lossy joint source-channel coding in the finite
  blocklength regime,'' \emph{IEEE Transactions on Information Theory},
  vol.~59, no.~5, pp. 2545--2575, 2013.

\bibitem{tian2010optimality}
C.~Tian, J.~Chen, S.~N. Diggavi, and S.~Shamai, ``Optimality and approximate
  optimality of source-channel separation in networks,'' in \emph{Information
  Theory Proceedings (ISIT), 2010 IEEE International Symposium on}.\hskip 1em
  plus 0.5em minus 0.4em\relax IEEE, 2010, pp. 495--499.

\bibitem{lomnitz2011communication}
Y.~Lomnitz and M.~Feder, ``Communication over individual channels,''
  \emph{Information Theory, IEEE Transactions on}, vol.~57, no.~11, pp.
  7333--7358, 2011.

\end{thebibliography}

\end{document}